\documentclass[sigconf]{acmart}

\usepackage{balance}
\usepackage{booktabs}
\usepackage[noend]{algpseudocode}
\usepackage{algorithmicx,algorithm}

\def\BibTeX{{\rm B\kern-.05em{\sc i\kern-.025em b}\kern-.08emT\kern-.1667em\lower.7ex\hbox{E}\kern-.125emX}}

\copyrightyear{2019}
\acmYear{2019}
\setcopyright{acmcopyright}
\acmConference[ISSAC '19] {2019 International Symposium on Symbolic and Algebraic Computation}{July 15--18, 2019}{Beijing, China}
\acmBooktitle{2019 Int'l Symposium on Symbolic $\&$ Algebraic Computation (ISSAC'19), July 15--18, 2019, Beijing, China}
\acmPrice{15.00}
\acmDOI{10.1145/3326229.3326234}
\acmISBN{978-1-4503-6084-5/19/07}

\begin{document}

\title[Averaging Method Revisited]{An Algorithmic Approach to Limit Cycles of Nonlinear Differential Systems: the Averaging Method Revisited}

\author{Bo Huang}
\affiliation{%
  \institution{LMIB-School of Mathematics and Systems Science,
Beihang University}
  \streetaddress{37 Xueyuan Road}
  \city{Beijing}
  \state{China}
  \postcode{100191}}
\affiliation{%
  \institution{Courant Institute of Mathematical Sciences,\\ New York University}
  \streetaddress{251 Mercer Street}
  \city{New York}
  \country{USA}
  \postcode{10012}}
\email{bohuang0407@buaa.edu.cn}

\author{Chee Yap}
\affiliation{%
  \institution{Courant Institute of Mathematical Sciences,\\ New York University}
  \streetaddress{251 Mercer Street}
  \city{New York}
  \country{USA}
  \postcode{10012}}
\email{yap@cs.nyu.edu}

%

%
\begin{abstract}
This paper introduces an algorithmic approach to the analysis of bifurcation of limit cycles from the centers of nonlinear continuous differential systems via the averaging method. We develop three algorithms to implement the averaging method. The first algorithm allows to transform the considered differential systems to the normal formal of averaging. Here, we restricted the unperturbed term of the normal form of averaging to be identically zero. The second algorithm is used to derive the computational formulae of the averaged functions at any order. The third algorithm is based on the first two algorithms that determines the exact expressions of the averaged functions for the considered differential systems. The proposed approach is implemented in Maple and its effectiveness is shown by several examples. Moreover, we report some incorrect results in published papers on the averaging method.
\end{abstract}

%
%


%

%

%
\maketitle

\section{Introduction}
Bounding the number of limit cycles for systems of polynomial differential equations is a long standing problem in the field of dynamical systems. As is well known, the second part of the 16th Hilbert's problem \cite{DH00,YI02} asks about ``the maximal number $H(n)$ and relative configurations of limit cycles'' for planar polynomial differential systems of degree $n$:
\begin{equation}\label{eq2.0}
\begin{split}
\dot{x}=f_n(x,y),\quad\dot{y}=g_n(x,y).
\end{split}
\end{equation}
Solving this problem, even in the case $n=2$, at the present state of knowledge seems to be hopeless. While it has not been possible to obtain uniform upper bounds for $H(n)$ in the near future, there has been success in finding lower bounds. Some known results are as follows: it is shown in \cite{lm79,s80} that $H(2)\geq4$ and $H(3)\geq13$ in \cite{ccj09}. In \cite{cn95}, it is proved that $H(n)$ grows at least as rapidly as $n^2\log n$. For the latest development about $H(n)$, we refer the reader to \cite{cl07,jl03}.

Recall that a limit cycle of system \eqref{eq2.0} is an isolated periodic orbit. It is the $\omega$-(forward) or $\alpha$-(backward) limit set of nearby orbits. One classical way of producing limit cycles is by perturbing a differential system which has a center. In this case the perturbed system displays limit cycles that bifurcate, either from the center (having the so-called Hopf bifurcation), or from some of the periodic orbits of the period annulus surrounding the center, see for instance Pontrjagin \cite{lsp34}, the book of Christopher-Li \cite{cl07}, and the hundreds of references quoted there.

In this paper we study the maximal number of limit cycles that bifurcate from the centers of the unperturbed systems (the so-called small-amplitude limit cycles). The main technique is based on the averaging method. We point out that the method of averaging is a classic and mature tool for studying isolated periodic solutions of nonlinear differential systems in the presence of a small parameter. The method has a long history that started with the classical works of Lagrange and Laplace, who provided an intuitive justification of the method. The first formalization of this theory was done in 1928 by Fatou. Important practical and theoretical contributions to the averaging method were made in the 1930s by Bogoliubov-Krylov, and in 1945 by Bogoliubov. The ideas of averaging method have extended in several directions for finite and infinite dimensional differentiable systems. We refer to the books of Sanders-Verhulst-Murdock \cite{svm07} and Llibre-Moeckel-Sim\'o \cite{jrc15} for a modern exposition of this subject.

We remark that most of these previous results developed the averaging method up to first order in a small parameter $\varepsilon$, and at most up to third order. In \cite{jmj13,jdm14} the averaging method at any order was developed to study isolated periodic solutions of nonsmooth but continuous differential systems. Recently, the averaging method has also been extended to study isolated periodic solutions of discontinuous differential systems; see \cite{jad15,jjd17}. In practice, the evaluation of the averaged functions is a computational problem that require powerful computerized resources. Moreover, the computational complexity grows very fast with the averaging order. In view of this, our objective in this paper is to present an algorithmic approach to develop the averaging method at any order and to further study periodic solutions of nonlinear continuous differential systems.

It is known that the Liapunov constants are a good tool for studying the number of small-amplitude limit cycles which can bifurcate from a singular point, i.e., a Hopf bifurcation. Over the years, a number of algorithms for efficient computation of Liapunov constants have been developed (see \cite{gt01,dw91,dw04} for instance). But a disadvantage of such an approach is that there is no clear geometry of the bifurcated limit cycles. In contrast, using the expressions of the averaged functions, we can estimate the size of the bifurcated limit cycles as a function of $\varepsilon$ for $|\varepsilon|>0$ sufficiently small, see \cite{jj07,rj17} for instance.

{\bf Overview of Paper.} The structure of our paper is as follows. In Section \ref{sect2}, we introduce the basic results on the averaging method for planar differential systems before presenting our main results in Section \ref{sect2.5}. We give our algorithms and briefly describe their implementation in Maple in Section \ref{sect4}. Its application is illustrated in Section \ref{sect5} using several examples including a cubic polynomial differential system known as \textit{Collins First Form} and a class of generalized Kukles polynomial differential systems of degree 6. We end with some discussions in Section \ref{sect6}.

In view of space limitation, we moved the proof of Theorem \ref{tt1} to Appendix \ref{A}. Two of the examples are found in Appendices \ref{B} and \ref{C}. The version with appendices may be found at our website
\href{https://cs.nyu.edu/exact/papers/}
{\textcolor{blue}{\underline{\url{https://cs.nyu.edu/exact/papers/}}}} as well as in the arXiv.

\section{Basic Theory of the Averaging Method} \label{sect2}
In this section we introduce the basic results on the averaging method that we shall use for studying the limit cycles which bifurcate from the centers of polynomial differential systems of degree $n_1$ in the form of
\begin{equation}\label{eq2.02}
\begin{split}
\dot{x}=P(x,y),\quad\dot{y}=Q(x,y).
\end{split}
\end{equation}
An accessible reference is \cite{chicone:ode:bk} (see also \cite{svm07}). The following definition is due to Poincar\'e (see \cite{jcm99}, Section 2).
\begin{definition}
We say that an isolated singular point $O$ of \eqref{eq2.02} is a \textit{center} if there exists a punctured neighbourhood $V$ of $O$, such that every orbit in $V$ is a cycle surrounding $O$.
\end{definition}
Without loss of generality we can assume that the center $O$ of system \eqref{eq2.02} is the origin of coordinates. In this case, after a linear change of variables and a rescaling of time variable, we can write system \eqref{eq2.02} in the form
\begin{equation}\label{eq2.03}
\begin{split}
\dot{x}&=\bar{P}_{\bar{\alpha}}(x,y)=-y+\sum_{m=2}^{n_1}P_m(x,y),\\
\dot{y}&=\bar{Q}_{\bar{\beta}}(x,y)=x+\sum_{m=2}^{n_1}Q_m(x,y),
\end{split}
\end{equation}
where $P_m$, $Q_m$ are homogeneous polynomials of degree $m$ in $x$ and $y$ with $\bar{\alpha}$ and $\bar{\beta}$ are parameters appearing as coefficients of $\bar{P}, \bar{Q}$ satisfying that system \eqref{eq2.03} has a center at the origin. It is well known since Poincar\'e \cite{HP88} and Liapunov \cite{mal47} that system \eqref{eq2.03} has a center at the origin if and only if there exists a local analytic first integral of the form $H(x,y)=x^2+y^2+F(x,y)$ defined in a neighborhood of the origin, where $F$ starts with terms of order higher than 2. For the well known center problem, see \cite{vd09,acj17}.

We now consider the perturbations of \eqref{eq2.03} of the form
\begin{equation}\label{eq2.04}
\begin{split}
\dot{x}&=\bar{P}_{\bar{\alpha}}(x,y)+p_{\alpha}(x,y,\varepsilon),\\
\dot{y}&=\bar{Q}_{\bar{\beta}}(x,y)+q_{\beta}(x,y,\varepsilon)
\end{split}
\end{equation}
with
\begin{equation}\label{eq2.04.1}
\begin{split}
p_{\alpha}(x,y,\varepsilon)=\sum_{j=1}^{k}\varepsilon^j\tilde{p}_j(x,y),\quad
q_{\beta}(x,y,\varepsilon)=\sum_{j=1}^{k}\varepsilon^j\tilde{q}_j(x,y),\nonumber
\end{split}
\end{equation}
where the polynomials $\tilde{p}_j, \tilde{q}_j$ are of degree at most $n_2$ (usually $n_2\geq n_1\geq2$) in $x$ and $y$ with $\alpha$ and $\beta$ are free parameters appearing as coefficients of $\tilde{p}_j, \tilde{q}_j$, and $\varepsilon$ is a small parameter. Note that by ``free parameters'' we mean that the coefficient of each monomial in $p_\alpha$ and $q_\beta$ is a distinct parameter in $\alpha$ or $\beta$. We are interested in the maximum number of small-amplitude limit cycles of \eqref{eq2.04} for $|\varepsilon|>0$ sufficiently small, which bifurcate at $\varepsilon=0$ from the center of \eqref{eq2.03}.

Usually, the averaging method deals with planar differential systems in the following normal form
\begin{equation}\label{eq2.01}
\begin{split}
\frac{dr}{d\theta}=\sum_{i=0}^k\varepsilon^iF_i(\theta,r)+\varepsilon^{k+1}R(\theta,r,\varepsilon),
\end{split}
\end{equation}
where $F_i: \mathbb{R}\times D\rightarrow\mathbb{R}$ for $i=0,1,\ldots,k$, and $R:\mathbb{R}\times D\times(-\varepsilon_0,\varepsilon_0)\rightarrow\mathbb{R}$ are $\mathcal{C}^{k}$ functions, $2\pi$-periodic in the first variable, being $D$ an open and bounded interval of $(0,\infty)$, and $\varepsilon_0$ is a small parameter. As one of the main hypotheses, it is assumed that the solution $\varphi(\theta,z)$ of the unperturbed differential system, $dr/d\theta=F_0(\theta,r)$, is $2\pi$-periodic in the variable $\theta$ for every initial condition $\varphi(0,z)=z\in D$.

The averaging method consists in defining a collection of functions $f_i: D\rightarrow\mathbb{R}$, called the $i$-th order averaged function, for $i=1,2,\ldots,k$, which control (their simple zeros control), for $\varepsilon$ sufficiently small, the isolated periodic solutions of the differential system \eqref{eq2.01}. In Llibre-Novaes-Teixeira \cite{jdm14} it has been established that
\begin{equation}\label{eq2.01.2}
\begin{split}
f_i(z)=\frac{y_i(2\pi,z)}{i!},
\end{split}
\end{equation}
where $y_i: \mathbb{R}\times D\rightarrow\mathbb{R}$, for $i=1,2,\ldots,k$, is defined recursively by the following integral equation
\begin{equation}\label{eq2.01.3}
\begin{split}
y_i(\theta,z)&=i!\int_0^{\theta}\Bigg[F_i(s,\varphi(s,z))+\sum_{\ell=1}^i\sum_{S_{\ell}}\frac{1}{b_1!b_2!2!^{b_2}\cdots b_{\ell}!\ell!^{b_{\ell}}}\\
&\quad\cdot\partial^LF_{i-\ell}(s,\varphi(s,z))\prod_{j=1}^{\ell}y_j(s,z)^{b_j}\Bigg]ds,
\end{split}
\end{equation}
where $S_{\ell}$ is the set of all $\ell$-tuples of non-negative integers $[b_1,b_2,\ldots,b_{\ell}]$ satisfying $b_1+2b_2+\cdots+\ell b_{\ell}={\ell}$ and $L=b_1+b_2+\cdots+b_{\ell}$. Here, $\partial^LF(\theta,r)$ denotes the Fr\'echet's derivative of order $L$ with respect to the variable $r$.

We remark that, in practical terms, the evaluation of the recurrence \eqref{eq2.01.3} is a computational problem. Recently in \cite{dn17} the Bell polynomials were used to provide a relatively simple alternative formula for the recurrence. In this paper, we will exploit this new formula in our algorithmic approach for solving this problem (see Section \ref{sect3.2}).

Related to the averaging functions \eqref{eq2.01.2} there exist two fundamentally different cases in \eqref{eq2.01}, namely, when $F_0=0$ and when $F_0\neq0$. We see that when $F_0\neq 0$, the formula for $y_i(\theta,z)$ in \eqref{eq2.01.3} requires the solution of a Cauchy problem because $y_i(\theta,z)$ appears on both sides of the equation (see Remark 3 in \cite{jdm14}). The investigation in this paper is restricted to the case where $F_0=0$. In this case, we have $\varphi(\theta,z)=z$ for each $\theta\in\mathbb{R}$. Then the integral in equation \eqref{eq2.01.3} simplifies to
\begin{equation}\label{eq3.3.0}
\begin{split}
y_1(\theta,z)&=\int_0^{\theta}F_1(s,z)ds,\\
y_i(\theta,z)&=i!\int_0^{\theta}\Bigg[F_i(s,z)+\sum_{\ell=1}^{i-1}\sum_{S_{\ell}}\frac{1}{b_1!b_2!2!^{b_2}\cdots b_{\ell}!\ell!^{b_{\ell}}}\\
&\quad\cdot\partial^LF_{i-\ell}(s,z)\prod_{j=1}^{\ell}y_j(s,z)^{b_j}\Bigg]ds.
\end{split}
\end{equation}
The following $k$-th order averaging theorem gives a criterion for the existence of limit cycles. Its proof can be found in Section 2 of \cite{jjd17}.
\begin{theorem}\cite{jjd17}\label{tt0}
Assume that $f_i\equiv0$ for $i=1,2,\ldots,j-1$ and $f_j\neq0$ with $j\in\{1,2,\ldots,k\}$. If there exists $\bar{r}\in D$ such that $f_j(\bar{r})=0$ and $f'_j(\bar{r})\neq0$, then for $|\varepsilon|>0$ sufficiently small, there exists a $2\pi$-periodic solution $r(\theta,\varepsilon)$ of \eqref{eq2.01} such that $r(0,\varepsilon)\rightarrow \bar{r}$ when $\varepsilon\rightarrow0$.
\end{theorem}

We remark that in order to analyze the Hopf bifurcation for system \eqref{eq2.04}, applying Theorem \ref{tt0}, we introduce a small parameter $\varepsilon$ doing the change of coordinates $x=\varepsilon X$, $y=\varepsilon Y$. After that we perform the polar change of coordinates $X=r\cos\theta$, $Y=r\sin\theta$, and by doing a Taylor expansion truncated at $k$-th order in $\varepsilon$ we obtain an expression for $dr/d\theta$ similar to \eqref{eq2.01} up to $k$-th order in $\varepsilon$. In doing so, the variable $\theta$ appears through sines and cosines, the differential equation in the form $dr/d\theta$ is $2\pi$-periodic. It suffices to take $D=\{r: 0<r<r_0\}$ with $r_0>0$ is arbitrary, since we restrict $F_0=0$, the unperturbed system has periodic solutions passing through the points $(0,r)$ with $0<r<r_0$.

In general, it is not an easy thing to determine the exact number of simple zeros of the averaged functions \eqref{eq2.01.2}, since the averaged functions may be too complicated, such as including square root functions, logarithmic functions, and the elliptic integrals. In the literature there is an abundance of papers dealing with zeros of the averaged functions (see for instance \cite{HJJ163,nt17,bh19} and references therein). The techniques and arguments to tackle this kind of problem are usually very long and technical.

As a summary of this section, we remark that, using the expressions of the averaged functions, one can estimate the size of bifurcated limit cycles. In fact we know that if the averaged function $f_j=0$ for $j=1,\ldots,k-1$ and $f_k\neq0$, and $\bar{r}$ is a simple zero of $f_k$, then by Theorem \ref{tt0} there is a limit cycle $r(\theta,\varepsilon)$ of the differential system \eqref{eq2.01} such that $r(0,\varepsilon)=\bar{r}+\mathcal{O}(\varepsilon)$. Then, going back through the changes of variables we have for the differential system $(\dot{X},\dot{Y})$ the limit cycle $(X(t,\varepsilon),Y(t,\varepsilon))=(\bar{r}\cos\theta,\bar{r}\sin\theta)+\mathcal{O}(\varepsilon)$. Now due to the scaling $x=\varepsilon X, y=\varepsilon Y$ the limit cycles that we find for the differential system \eqref{eq2.01} coming from our system \eqref{eq2.04}, are in fact limit cycles of the form $(x(t,\varepsilon),y(t,\varepsilon))=\varepsilon(\bar{r}\cos\theta,\bar{r}\sin\theta)+\mathcal{O}(\varepsilon^2)$ for system \eqref{eq2.04}, which tends to the origin from the origin, i.e., are limit cycles coming by a Hopf bifurcation, for more details on these kind of bifurcations see \cite{yk04} for instance.

\section{Main Results}\label{sect2.5}
Denote the exact upper bound for the number of positive simple zeros of the $i$-th order averaged function $f_i(r)$ associated to system \eqref{eq2.04} by $H_i(n_1,n_2)$ for $i=1,\ldots,k$. Applying Theorem \ref{tt0}, we know that the maximal number of small-amplitude limit cycles of \eqref{eq2.04} is $H_i(n_1,n_2)$ and this number can be reached. In this work, we attempt to prove upper bounds on the number of zeros of the $k$-th order averaged function. Our main theorem is the following:
\begin{theorem}\label{tt1}
Assume that $F_0=0$ in the normal form \eqref{eq2.01} associated to the system \eqref{eq2.04}, then there exist a non-negative integer $\nu_i\leq i-1$ and a polynomial function $\bar{f}_i(r)=\sum_{j=0}^{N_i}c_jr^j$ with $N_i\leq in_2$, such that $r^{\nu_i}f_i(r)=\bar{f}_i(r)$ for $i=1,\ldots,k$, where the coefficients $c_j\in\mathbb{Q}[\pi]$ with degree no more than $i$ in $\pi$.
\end{theorem}

A detailed proof of it can be found in Appendix \ref{A}. This result is the first work that deals with the bifurcation of limit cycles of system \eqref{eq2.04} in the general class of perturbations (see \cite{jj15,jm16} for a few results on some systems of special form). This theorem tells us that the maximum number of small-amplitude limit cycles of \eqref{eq2.04}, which bifurcate from the center of \eqref{eq2.03} is always finite ( $H_k(n_1,n_2)\leq N_k$). But, for a given system \eqref{eq2.04}, how can we determine the exact value of $N_i$ for $i=1,\ldots,k$? In this paper, we provide an algorithmic approach to the solution (see Algorithm 3 in Section \ref{sect3.2}).

Applying Theorems \ref{tt0} and \ref{tt1}, we obtain the Theorem \ref{tt2} on $f_k(r)$. We first introduce some notations based on Theorems \ref{tt0} and \ref{tt1} before we state this result. Let $R^*$ be the real polynomial ring $\mathbb{Q}[\bar{\alpha},\bar{\beta},\alpha,\beta]$. Then for each $\bar{f}_i(r)\in R^*[\pi][r]$, we define $\mbox{coeffs}(\bar{f}_i;r,\pi)=\{c_{j_{1},{j_2}}: j_1=0,\ldots,N_i; j_2=0,\ldots,i\}$, where
\[\bar{f}_i(r)=\sum_{j_2=0}^i\sum_{j_1=0}^{N_i}c_{j_{1},{j_2}}r^{j_1}\pi^{j_2}.\]
Then $\Sigma_k=\cup_{i=1}^{k-1}\mbox{coeffs}(f_i;r,\pi)=0\subseteq R^*$. Now taking the above notations into account and applying Theorems \ref{tt0} and \ref{tt1}, we obtain the following theorem on $f_k(r)$.
\begin{theorem}\label{tt2}
Assume that $F_0=0$ in the normal form \eqref{eq2.01} associated to the system \eqref{eq2.04}. Then there exist non-negative integers $\tilde{\nu}_k\leq k-1$, $\tilde{N}_k\leq kn_2$ and a polynomial function $\tilde{f}_k(r)=\sum_{j=0}^{N_k}\tilde{c}_jr^j$, such that $f_k(r)$ in \eqref{eq2.01.2} has the form $r^{\tilde{\nu}_k}f_k(r)=\tilde{f}_k(r)$ and the coefficients $$\tilde{c}_j\in\mathbb{Q}[\pi,\bar{\alpha},\bar{\beta},\alpha,\beta]/\Sigma_k$$
with degree no more than $k$ in $\pi$.
\end{theorem}
\begin{proof}
The conclusion follows directly from the conditions $\bar{f}_1=\bar{f}_2=\cdots=\bar{f}_{k-1}=0$.
\end{proof}

We remark that, the study of the number of zeros of $f_k(r)$ is currently not-algorithmic. Below we give our analysis on this.

Let $\bar{N}=|\bar{\alpha}|+|\bar{\beta}|+|\alpha|+|\beta|$ be the number of parameters in system \eqref{eq2.04}, and $V(\Sigma_k)\subseteq\mathbb{R}^{\bar{N}}$ is the variety defined by $\Sigma_k$. For any point $p^*\in V(\Sigma_k)$, let $\bar{f}_k(r;p^*)\in\mathbb{R}[r]$ be the real polynomial when the parameter are instantiated by $p^*$. Finally let $\#(p^*)$ denote the maximal number of zeros (counted with multiplicity) of $\bar{f}(r;p^*)$. It follows that $H_k(n_1,n_2)\leq\mbox{max}\{\#(p^*: p\in V(\Sigma_k))\}$.

In order to study the number of zeros of function $f_k(r)$, according to our Theorem \ref{tt1}, it suffices to consider the number of zeros of a polynomial function. Here we provide the \textit{Descartes theorem} (see \cite{in65}) to obtain the upper bound of the number of zeros for the polynomial functions.

\begin{lemma}\label{lem0}
(Descartes theorem). Consider the real polynomial $m(x)=a_{s_1}x^{s_1}+a_{s_2}x^{s_2}+\cdots+a_{s_m}x^{s_m}$ with $0=s_1<s_2<\cdots<s_m$ and $a_{s_j}\neq0$ real constants for $j\in\{1,2,\ldots,m\}$. When $a_{s_j}a_{s_{j+1}}<0$, we say that $a_{s_j}$ and $a_{s_{j+1}}$ have a variation of sign. If the number of variations of signs is $m^*$, then $m(x)$ has at most $m^*$ positive real roots. Moreover, it is always possible to choose the coefficients of $m(x)$ in such a way that $m(x)$ has exactly $m-1$ positive real roots.
\end{lemma}

\section{Algorithms for the $k$-th Order Averaging Theorem}\label{sect4}

In this section we will provide an algorithmic approach to revisit the averaging method. According to the averaging method described in Section 2, it is necessary to take the following steps to study the bifurcation of limit cycles for system \eqref{eq2.04}.

{\bf STEP 1}. Write the perturbed system \eqref{eq2.04} in the normal form of averaging \eqref{eq2.01} up to $k$-th order in $\varepsilon$.

{\bf STEP 2}. (i) Compute the exact formula for the $k$-th order integral function $y_k(\theta,z)$ in \eqref{eq3.3.0}. (ii) Derive the symbolic expression of the $k$-th order averaged function $f_k(z)$ by \eqref{eq2.01.2}.

{\bf STEP 3}. Determine the exact upper bound for number of positive simple zeros of $f_k(z)$.

In the following subsections we will present algorithms to implement the first two steps. We use ``Maple-like'' pseudo-code, based on our Maple implementation. Using these algorithms we reduce the problem of studying the number of limit cycles of system \eqref{eq2.04} to the problem of detecting {\bf STEP 3}.

\subsection{Algorithm for {\bf STEP 1}}\label{sect3.1}
In this subsection we will devise an efficient algorithm which can be used to transform system \eqref{eq2.04} into the form \eqref{eq2.01}. Our algorithm can derive \eqref{eq2.01} at any order in $\varepsilon$.

Now refer to \eqref{eq2.04}, making the change of variables $x=\varepsilon\cdot r\cdot C$ and $y=\varepsilon\cdot r\cdot S$ with $C=\cos\theta$ and $S=\sin\theta$, we present the algorithm {\bf Normalize} below based on the above analysis.
\begin{small}
\begin{algorithm}[H]
\caption{{\bf Normalize}$(\bar{P}_{\bar{\alpha}},\bar{Q}_{\bar{\beta}},p_{\mathbb{\alpha}},q_{\mathbb{\beta}},k)$}
\hspace*{0.02in} {\bf Input:}
a perturbed system \eqref{eq2.04} with a order $k\geq1$\\
\hspace*{0.02in} {\bf Output:}
an expression for $dr/d\theta$ similar to \eqref{eq2.01} up to $k$-th order in $\varepsilon$
\begin{algorithmic}[1]
\State $dX:=\mbox{normal}(\mbox{subs}(x=\varepsilon X,y=\varepsilon Y,\bar{P}_{\bar{\alpha}}+p_{\mathbb{\alpha}})/\varepsilon)$;
\State $dY:=\mbox{normal}(\mbox{subs}(x=\varepsilon X,y=\varepsilon Y,\bar{Q}_{\bar{\beta}}+q_{\mathbb{\beta}})/\varepsilon)$;
\State $R0:=\mbox{normal}\left(\mbox{subs}\left(X=r\cdot C,Y=r\cdot S,\frac{r\cdot(C\cdot dX+S\cdot dY)}{C\cdot dY-S\cdot dX}\right)\right)$;
\State $T:=\mbox{taylor}(R0,\varepsilon=0,k+1)$;
\State $H:=\mbox{expand}(\mbox{convert}\left(T,\mbox{polynom}\right))$;
\If{ $\mbox{coeff}(\varepsilon\cdot H,\varepsilon)=0$ }
\For{$i$ {\bf from} 1 {\bf to} $k$}
\State $f_i:=\mbox{coeff}(H,\varepsilon^i)$;
\State $F_{i,1}:=\mbox{prem}\left(\mbox{numer}(f_i),C^2+S^2-1,C\right)$;
\State
$F_{i,2}:=\mbox{prem}\left(\mbox{denom}(f_i),C^2+S^2-1,C\right)$;
\State
$F_i:=\mbox{normal}({F_{i,1}}/{F_{i,2}})$;
\EndFor
\EndIf
\State $dr/d\theta:=\mbox{subs}(C=\cos\theta,S=\sin\theta,\sum_{j=1}^kF_j\varepsilon^j$);
\State \Return $dr/d\theta$;
\end{algorithmic}
\end{algorithm}
\end{small}
The {\bf if} hypothesis in line 6 is to make sure that $F_0=0$. In line 9 the function $\mbox{prem}(a,b,x)$ is the pseudo-remainder of $a$ with respect to $b$ in the variable $x$. The following lemma is obtained directly by the property of the pseudo-remainder.

\begin{lemma}
The expressions $F_{i,j}$ with $j\in\{1,2\}$ in the algorithm {\bf Normalize} have the following properties:
\[F_{i,j}=f_{i,j}(r,S)C+g_{i,j}(r,S),\]
where $f_{i,j}$ and $g_{i,j}$ are polynomials in the variables $r$ and $S$.
\end{lemma}

\subsection{Algorithms for {\bf STEP 2}}\label{sect3.2}
This subsection is devoted to provide effective algorithms to compute the formula and exact expression of the $k$-th order averaged function.

According to \eqref{eq3.3.0}, we should take the following substeps to compute the $k$-th order averaged function of system \eqref{eq2.01}:

{\bf Substep 1}. Compute the exact formula for the $k$-th order integral function $y_k(\theta,z)$.

{\bf Substep 2}. Output the symbolic expression for the $k$-th order averaged function $f_k(r)$ (not simplified by using $f_1\equiv f_2\equiv\cdots\equiv f_{k-1}\equiv0$) for a given differential system \eqref{eq2.04}.

We first recall the partial Bell polynomials which can be used to implement the first substep. For $\ell$ and $m$ positive integers, the Bell polynomials:
\begin{equation}\label{e3.2.1}
\begin{split}
B_{\ell,m}(x_1,\ldots,x_{\ell-m+1})=\sum_{\tilde{S}_{\ell,m}}\frac{\ell!}{b_1!b_2!\cdots b_{\ell-m+1}!}\prod_{j=1}^{\ell-m+1}\left(\frac{x_j}{j!}\right)^{b_j},\nonumber
\end{split}
\end{equation}
where $\tilde{S}_{\ell,m}$ is the set of all $(\ell-m+1)$-tuples of nonnegative integers $[b_1,b_2,\ldots,b_{\ell-m+1}]$ satisfying $b_1+2b_2+\cdots+(\ell-m+1)b_{\ell-m+1}=\ell$, and $b_1+b_2+\cdots+b_{\ell-m+1}=m$.

Then the integral in equation \eqref{eq3.3.0} reads (\cite{dn17}, Theorem 2)
\begin{equation}\label{e3.2.2}
\begin{split}
y_1(\theta,z)&=\int_0^{\theta}F_1(s,z)ds,\\
y_i(\theta,z)&=i!\int_0^{\theta}\Bigg[F_i(s,z)+\sum_{\ell=1}^{i-1}\sum_{m=1}^{\ell}\frac{1}{ \ell!}\partial^mF_{i-\ell}(s,z)\\
&\quad\cdot B_{\ell,m}(y_1(s,z),\ldots,y_{\ell-m+1}(s,z))\Bigg]ds.
\end{split}
\end{equation}

The algorithm {\bf Averformula}, presented below, is based on \eqref{e3.2.2} that can be used to derive the formula of the $k$-th order integral function $y_k(\theta,z)$ ({\bf Substep 1}).

\begin{algorithm}[H]
\caption{{\bf Averformula}$(k)$}
\hspace*{0.02in} {\bf Input:}
a order $k\geq1$ of the normal form \eqref{eq2.01}\\
\hspace*{0.02in} {\bf Output:}
a set of formulae $Y_k$ associated to the integral function $y_k(\theta,z)$
\begin{algorithmic}[1]
\State $\mbox{SU}:=0$; $\mbox{TU}:=0$;
\For{$\ell$ {\bf from} 1 {\bf to} $k-1$}
\For{$m$ {\bf from} 1 {\bf to} $\ell$}
\State $\mbox{SU}:=\mbox{SU}+\frac{1}{\ell!}\cdot\mbox{Diff}
(F_{k-\ell}(s,z),z\$m)\cdot\mbox{{\bf IncompleteBellB}}(\ell,m,y_1(s,z),\ldots,y_{\ell-m+1}(s,z))$;
\State $\mbox{TU}:=\mbox{TU}+\frac{1}{\ell!}\cdot\mbox{Diff}
(F_{k-\ell},r\$m)\cdot\mbox{{\bf IncompleteBellB}}(\ell,m,y_1,\ldots,y_{\ell-m+1})$;
\EndFor
\EndFor
\State
\begin{small}
$Y_k:=\Big\{\int_0^{\theta}k!\cdot\big(F_k(s,z)+SU\big)ds,\big[ \int_0^{\theta}k!\cdot\big(F_k+TU\big)d\theta,\int_0^{2\pi}\big(F_k+TU\big)d\theta\big]\Big\}$;
\end{small}
\State \Return $Y_k$;
\end{algorithmic}
\end{algorithm}

For the generation of the Bell polynomials (lines 4 and 5) we use the routine \textit{IncompleteBellB} built-in Maple. We give the outputs $Y_k$ of the algorithm for $k=1,2$ (see \eqref{B01} in Appendix \ref{B}). Note that the formula of $y_k(\theta,z)$ is the first element in the set $Y_k$. The second element in $Y_k$ (where $F_k$ without the dependence on $(s,z)$) can be used to derive an exact expression of $f_k$ if we give a concrete differential system \eqref{eq2.04} (then $F_k$ can be assigned to values by the algorithm {\bf Normalize}), see next algorithm {\bf AverFun}. We also remark that the formula for the $k$-th order averaged function $f_k$ can be obtained directly from $y_k$ using \eqref{eq2.01.2}, so we omit the formula for $f_k$ in our algorithm {\bf Averformula}. We deduce explicitly the formulae of $y_k$'s up to $k=5$ (see \eqref{BB1} in Appendix \ref{B}); one can verify that the outputs of our algorithm are consistent with the results given in \cite{jdm14,jjd17}. In fact our algorithm can compute arbitrarily high order formulae of $y_k$'s. In Section 4, we will study a cubic differential system (\textit{Collins First Form}) and a class of generalized Kukles systems to show the feasibility of our algorithm.

In the last subsection, we provide an algorithm {\bf Normalize} to transform system \eqref{eq2.04} into the form of $dr/d\theta$ (normal form of averaging). The algorithm {\bf AverFun}, presented below, is based on the algorithms {\bf Normalize} and {\bf Averformula}, which provides a straightforward calculation method to derive the exact expression of the $k$-th order averaged function for a given differential system in the form \eqref{eq2.04} ({\bf Substep 2}).
\begin{small}
\begin{algorithm}[H]
\caption{{\bf AverFun}$(\bar{P}_{\bar{\alpha}},\bar{Q}_{\bar{\beta}},p_{\mathbb{\alpha}},q_{\mathbb{\beta}},k)$}
\hspace*{0.02in} {\bf Input:}
a perturbed system \eqref{eq2.04} with a order $k\geq1$\\
\hspace*{0.02in} {\bf Output:}
an expression of the $k$-th order averaged function $f_k$ of $dr/d\theta$
\begin{algorithmic}[1]
\State $dr/d\theta:=\mbox{\bf Normalize}(\bar{P}_{\bar{\alpha}},\bar{Q}_{\bar{\beta}},p_{\mathbb{\alpha}},q_{\mathbb{\beta}},k)$;
\For{$h$ {\bf from} 1 {\bf to} $k$}
\State $F_h:=\mbox{coeff}(dr/d\theta,\varepsilon^h)$;
\State $Y_h:=\mbox{{\bf Averformula}}(h)$;
\State $y_h:=\mbox{value}(\mbox{op}(1,\mbox{op}(2,Y_h)))$;
\State $f_h:=\mbox{factor}(\mbox{value}(\mbox{op}(2,\mbox{op}(2,Y_h))))$;
\EndFor
\State \Return $f_k$;
\end{algorithmic}
\end{algorithm}
\end{small}
According to our Theorem \ref{tt1}, we know that the output of the algorithm {\bf AverFun} has the property $f_k\in\mathbb{Q}[\pi][r,r^{-1}]$. And the numerator of the expression $f_k$ is a polynomial function with degree $N_k$. In practice, the calculation of $f_k$ typically requires powerful computer resources as the computational complexity grows exponentially with order $k$.  It turns out that we can greatly improve the speed by updating the obtained $dr/d\theta$ by using the conditions $f_1\equiv f_2\equiv\cdots\equiv f_{k-1}\equiv0$.

We implemented all the algorithms presented in this section in Maple. In the next section, we will apply our general algorithmic approach to analyze the bifurcation of limit cycles for several concrete differential systems in order to show its feasibility.

\section{Experiments}\label{sect5}
In this section, we present the bifurcation of limit cycles for a cubic polynomial differential system as an illustration of our approach explained above. In addition, the bifurcation of limit cycles from the centers of a class of generalized Kukles polynomial differential systems of degree 6 is studied when it is perturbed inside the class of all polynomial differential systems of the same degree, and as an application of our method, we also report some results on quadratic differential systems with isochronous centers. The obtained results of our experiments show the feasibility of our approach.

\subsection{Illustrative Example}\label{sect4.0}
In this subsection, we consider a cubic center of the following polynomial system
\begin{equation}\label{e4.1}
\begin{split}
\dot{x}=-y+x^2y,\quad\dot{y}=x+xy^2.
\end{split}
\end{equation}
This system is known as \textit{Collins First Form}, see \cite{jj15} for more details.

More concretely, we consider the perturbations of \eqref{e4.1} in the form of
\begin{equation}\label{e4.2}
\begin{split}
\dot{x}&=-y+x^2y+\sum_{s=1}^7\varepsilon^sp_s(x,y),\\
\dot{y}&=x+xy^2+\sum_{s=1}^7\varepsilon^sq_s(x,y),
\end{split}
\end{equation}
where
\begin{equation}\label{e4.3}
\begin{split}
p_s(x,y)&=\alpha_{s,1}x+\alpha_{s,2}y+\alpha_{s,3}x^2+\alpha_{s,4}xy+\alpha_{s,5}y^2
+\alpha_{s,6}x^3\\
&\quad+\alpha_{s,7}x^2y+\alpha_{s,8}xy^2+\alpha_{s,9}y^3,\\
q_s(x,y)&=\beta_{s,1}x+\beta_{s,2}y+\beta_{s,3}x^2+\beta_{s,4}xy+\beta_{s,5}y^2
+\beta_{s,6}x^3\\
&\quad+\beta_{s,7}x^2y+\beta_{s,8}xy^2+\beta_{s,9}y^3,\nonumber
\end{split}
\end{equation}
being $\alpha_{s,j}$ and $\beta_{s,j}$, for $s=1,\ldots,7$ and $j=1,\ldots,9$, real constants.

Next, we use our algorithms to study the maximum number of limit cycles of \eqref{e4.2} that bifurcate from the center of \eqref{e4.1}. Applying our algorithm {\bf Normalize} by taking $k=7$ we obtain
\begin{equation}\label{e4.4}
\begin{split}
\frac{d r}{d\theta}=\sum_{i=1}^7\varepsilon^i F_i(\theta,r)+\mathcal{O}(\varepsilon^8).
\end{split}
\end{equation}
Here we give only the expression of $F_1(\theta,r)$, the explicit expressions of $F_i(\theta,r)$ for $i=2,\ldots,7$ are quite large so we omit them.
\[F_1(\theta,r)=r(\alpha_{1,2}+\beta_{1,1})SC+r(-\alpha_{1,1}+\beta_{1,2})S^2+r\alpha_{1,1}\]
with $C=\cos\theta$ and $S=\sin\theta$.

Using our algorithm {\bf AverFun} in Section \ref{sect4} and computing $f_1$ we obtain $f_1(r)=\pi r(\alpha_{1,1}+\beta_{1,2})$. Clearly equation $f_1(r)$ has no positive zeros. Thus the first averaged function does not provide any information about the limit cycles that bifurcate from the center of \eqref{e4.1} when we perturb it.

Computing $f_2$ we obtain
\begin{equation}\label{ef2}
\begin{split}
f_2(r)&=\frac{\pi r}{2}\big(\pi\alpha_{1,1}^2+2\pi\alpha_{1,1}\beta_{1,2}+\pi\beta_{1,2}^2
+\alpha_{1,1}\alpha_{1,2}-\alpha_{1,1}\beta_{1,1}\\
&\quad+\alpha_{1,2}\beta_{1,2}-\beta_{1,1}\beta_{1,2}+2\alpha_{2,1}+2\beta_{2,2}\big).\nonumber
\end{split}
\end{equation}

According to our Theorem \ref{tt1}, we take $\bar{f}_2(r)=f_2(r)$ with degree $N_2=1$, and $c_1$ is a polynomial in $\pi$ with degree 2. Note that $f_1(r)=0$ means that $\beta_{1,2}=-\alpha_{1,1}$. Using this condition we can simplify $f_2(r)$ into the form $f_2(r)=\pi r(\alpha_{2,1}+\beta_{2,2})$. As for the first averaged function, the second one also does not provide information on the bifurcating limit cycles. From now on, for each $k=3,\ldots,7$, we will perform the calculation of the averaged function $f_k$ under the hypothesis $f_j\equiv0$ for $j=1,\ldots,k-1$.

Doing $\beta_{2,2}=-\alpha_{2,1}$ and computing $f_3$ we obtain
\[f_3(r)=\frac{1}{4}\pi r\left(A_2r^2+A_0\right),\]
where
\[A_2=4\alpha_{1,1}+3\alpha_{1,6}+\alpha_{1,8}+\beta_{1,7}+3\beta_{1,9},\quad A_0=4(\alpha_{3,1}+\beta_{3,2}).\]
Therefore $f_3(r)$ can have at most one positive real root. From Theorem \ref{tt0} it follows that the 3-th order averaging provides the existence of at most one small-amplitude limit cycle of system \eqref{e4.2} and this number can be reached by Lemma \ref{lem0}, since $A_i$ for $i=0,2$ are independent constants ($\partial(A_2,A_0)/\partial(\beta_{1,7},\beta_{3,2})=4\neq0$).

To consider the 4-th order averaging theorem we take $\beta_{1,7}=-A_2+\beta_{1,7}$ and $\beta_{3,2}=-A_0/4+\beta_{3,2}$. Computing $f_4$ we obtain
\[f_4(r)=\frac{1}{4}\pi r\left(B_2r^2+B_0\right),\]
where
\begin{equation}\label{e4.6}
\begin{split}
B_2&=4\alpha_{1,1}\alpha_{1,2}+2\alpha_{1,1}\alpha_{1,7}+2\alpha_{1,1}\beta_{1,8}+\alpha_{1,2}\alpha_{1,8}\\
&\quad+3\alpha_{1,2}\beta_{1,9}+\alpha_{1,3}\alpha_{1,4}-2\alpha_{1,3}\beta_{1,3}+\alpha_{1,4}\alpha_{1,5}\\
&\quad+2\alpha_{1,5}\beta_{1,5}+\alpha_{1,8}\beta_{1,1}+3\beta_{1,1}\beta_{1,9}-\beta_{1,3}\beta_{1,4}\\
&\quad-\beta_{1,4}\beta_{1,5}+4\alpha_{2,1}+3\alpha_{2,6}+\alpha_{2,8}+\beta_{2,7}+3\beta_{2,9},\\
B_0&=4(\alpha_{4,1}+\beta_{4,2}).\nonumber
\end{split}
\end{equation}
It is obvious that $f_4(r)$ can have at most one positive real root. From Theorem \ref{tt0} it follows that the 4-th order averaging provides the existence of at most one small-amplitude limit cycle of system \eqref{e4.2} and this number can be reached ($B_2$ and $B_0$ are independent constants).

Letting $\beta_{2,7}=-B_2+\beta_{2,7}$ and $\beta_{4,2}=-B_0/4+\beta_{4,2}$ we obtain $f_4(r)=0$. Computing $f_5$ we obtain
\[f_5(r)=\frac{1}{4}\pi r\left(C_4r^4+C_2r^2+C_0\right),\]
where
\begin{equation}
\begin{split}
C_4=2\alpha_{1,1}+2\alpha_{1,6}+\alpha_{1,8}+\beta_{1,9},\quad C_0=4(\alpha_{5,1}+\beta_{5,2}).\nonumber
\end{split}
\end{equation}
We do not explicitly provide the expression of $C_2$, because it is very long. It is not hard to check that $C_4$, $C_2$ and $C_0$ are independent constants. Therefore $f_5(r)$ can have at most two positive real roots. Then the 5-th order averaging provides the existence of at most two small-amplitude limit cycle of system \eqref{e4.2} and this number can be reached.

To consider the 6-th order averaging theorem we let $\beta_{1,9}=-C_4+\beta_{1,9}$, $\beta_{3,7}=-C_2+\beta_{3,7}$ and $\beta_{5,2}=-C_0/4+\beta_{5,2}$. Computing $f_6$ we obtain
\[f_6(r)=\frac{1}{24}\pi r\left(D_4r^4+D_2r^2+D_0\right),\]
where
\begin{equation}
\begin{split}
D_4&=12\alpha_{1,1}\alpha_{1,7}-6\alpha_{1,1}\alpha_{1,9}
-12\alpha_{1,1}\beta_{1,1}-18\alpha_{1,1}\beta_{1,6}\\
&\quad-12\alpha_{1,2}\alpha_{1,6}+7\alpha_{1,3}\alpha_{1,4}
-18\alpha_{1,3}\beta_{1,3}-20\alpha_{1,3}\beta_{1,5}\\
&\quad+7\alpha_{1,4}\alpha_{1,5}-4\alpha_{1,5}\beta_{1,3}
-6\alpha_{1,5}\beta_{1,5}-18\alpha_{1,6}\alpha_{1,9}\\
&\quad-12\alpha_{1,6}\beta_{1,1}-18\alpha_{1,6}\beta_{1,6}
-6\alpha_{1,8}\alpha_{1,9}-6\alpha_{1,8}\beta_{1,6}\\
&\quad+\beta_{1,3}\beta_{1,4}+\beta_{1,4}\beta_{1,5}
+12\alpha_{2,1}+12\alpha_{2,6}+6\alpha_{2,8}+6\beta_{2,9},\\
D_0&=24(\alpha_{6,1}+\beta_{6,2}).\nonumber
\end{split}
\end{equation}
Here we do not provide the explicit expression of $D_2$ because it is quite long. Moreover $D_4$, $D_2$ and $D_0$ are independent constants. In fact only $D_4$ presents the parameter $\alpha_{2,6}$, only $D_2$ has the parameter $\alpha_{2,2}$, and $D_0$ is the only one with parameters $\alpha_{6,1}$ and $\beta_{6,2}$. Hence $f_6(r)$ has at most two positive simple roots. Then the 6-th order averaging provides the existence of at most two small-amplitude limit cycle of system \eqref{e4.2} and this number can be reached.

To consider the 7-th order averaging theorem we take $\beta_{2,9}=-D_4/6+\beta_{2,9}$, $\beta_{4,7}=-D_2/6+\beta_{4,7}$ and $\beta_{6,2}=-D_0/24+\beta_{6,2}$. Computing $f_7$ we obtain
\[f_7(r)=-\frac{1}{48}\pi r\left(E_6r^6+E_4r^4+E_2r^2+E_0\right),\]
where
\begin{equation}
\begin{split}
E_6=-3(\alpha_{1,1}+\alpha_{1,6}+\alpha_{1,8}),\quad
E_0=-48(\alpha_{7,1}+\beta_{7,2}).\nonumber
\end{split}
\end{equation}
Here we do not provide the explicit expressions of $E_2$ and $E_4$ because they are quite long. Moreover $E_j$ for $j=0,2,4,6$ are independent constants. Hence $f_7(r)$ has at most three positive simple roots. Then the 7-th order averaging provides the existence of at most three small-amplitude limit cycle of system \eqref{e4.2} and this number can be reached.

We remark that our averaged functions $f_j(r)$ for $j=1,\ldots,5$ are consistent with the forms in \cite{jj15}. However, our averaged function $f_6(r)$ looks much simpler than the form given in \cite{jj15}, this is because we rigorously simplify the function $f_6(r)$ under the conditions $f_1\equiv f_2\equiv\cdots\equiv f_5\equiv0$. The averaged function $f_6(r)$ in \cite{jj15} is not correct, because the authors do not simplify this expression in a right way (in fact one should note that the isolated parameter $\beta_{3,7}$ contains the parameter $\beta_{1,9}$). As a consequence, the maximum number $\{3\}$ of limit cycles of system \eqref{e4.1} up to the 6-th order averaging they obtained can not be reached. Thus, some calculations of the averaged functions in \cite{jj15} need to be reconsidered algorithmically, using the algorithm and exact formula of the averaged function in this paper.

Here we restate the result related to the \textit{Collins First Form} as follows.
\begin{theorem}\label{tt4.1}
For $|\varepsilon|>0$ sufficiently small the maximum number of small-amplitude limit cycles of system \eqref{e4.2} is 3 using the 7-th order averaging method, and this number can be reached.
\end{theorem}

\subsection{A Class of Generalized Kukles Differential Systems}\label{sect4.1}
In this subsection we consider the perturbations
\begin{equation}\label{eq4.2.1}
\begin{split}
\dot{x}&=-y+\sum_{s=1}^{6}\sum_{j=0}^6\sum_{i=0}^j\varepsilon^sa_{s,j,i}x^{j-i}y^i,\\
\dot{y}&=x+ax^5y+bx^3y^3+cxy^5+\sum_{s=1}^{6}\sum_{j=0}^6\sum_{i=0}^j\varepsilon^sb_{s,j,i}x^{j-i}y^i
\end{split}
\end{equation}
of system $\eqref{eq4.2.1}_{\varepsilon=0}$, where $a_{s,j,i}$ and $b_{s,j,i}$ are real parameters, for $s=1,\ldots,6$, $0\leq i\leq j\leq6$, and $a, b, c$ are real coefficients satisfying $a^2+b^2+c^2\neq0$. We note that the bifurcation of limit cycles of $\eqref{eq4.2.1}$ has been studied in \cite{jm16} up to 6-th order averaging theorem (\cite{jm16}, Section 7.3). Here restudy it by using our algorithmic approach to illustrate its feasibility.

We remark that taking $k=6$ our algorithm {\bf Normalize} can not  pass the {\bf if} hypothesis, this is because the unperturbed term (i.e., the constant term of $H$ in the algorithm {\bf Normalize}) \[F_0=\frac{r(a_{1,0,0}C+b_{1,0,0}S)}{r-a_{1,0,0}S+b_{1,0,0}C}\Big|_{C=\cos\theta,S=\sin\theta}\]
does not vanish. So we have to exclude the perturbed terms $\varepsilon a_{1,0,0}$ and $\varepsilon b_{1,0,0}$ in $\eqref{eq4.2.1}$. However, the authors in \cite{jm16} obtained a wrong expression of $F_0$ in the form
\[F_0=\frac{ra_{1,0,0}(C+S)}{r+a_{1,0,0}(C-S)}\Big|_{C=\cos\theta,S=\sin\theta}.\]
In fact one can easily check this mistake by manual calculation. So the calculations of the averaged functions of system $\eqref{eq4.2.1}$ in \cite{jm16} must be redone.

Now consider system $\eqref{eq4.2.1}$, letting $a_{1,0,0}=b_{1,0,0}=0$ and using our algorithm {\bf AverFun} in Section \ref{sect4} we obtain the averaged functions up to 6-th order as follows. Since the calculations and arguments are quite similar to those used in the previous subsection we do not explicitly present the process here.
\begin{equation}\label{eq4.2.3}
\begin{split}
f_1(r)&=\pi r(a_{1,1,0}+b_{1,1,1}),\quad f_2(r)=\pi r(a_{2,1,0}+b_{2,1,1}),\\
f_3(r)&=\frac{1}{4}\pi r(E_2r^2+E_0),\quad f_4(r)=\frac{1}{4}\pi r(G_2r^2+G_0),\\
f_5(r)&=\frac{1}{8}\pi r(H_4r^4+H_2r^2+H_0),\\
f_6(r)&=-\frac{1}{24}\pi r(I_4r^4+I_2r^2+I_0).
\end{split}
\end{equation}
The expressions of $E_i,G_i$ for $i=0,2$ and $H_j,I_j$ for $j=0,2,4$ are quite long so we omit them for brevity.


In view of these expressions in \eqref{eq4.2.3}, we verified that (Theorem 3 in \cite{jm16}) the averaging theorem up to sixth order provides the existence of at most two small-amplitude limit cycles of system \eqref{eq4.2.1}.


\subsection{Quadratic Systems}\label{sect4.3}
In order to save space, we put the results in Appendix \ref{C}.

\section{Discussions}\label{sect6}
In this paper we present a systematical approach to study the maximum number of limit cycles of differential system \eqref{eq2.04} for $|\varepsilon|>0$ sufficiently small, which bifurcate from the centers of differential systems in the form of \eqref{eq2.03}. In general, we give three algorithms to analyze the averaging method. Then with the aid of these algorithms, we reduce the study of the number of limit cycles of system \eqref{eq2.04} to the problem of estimating the number of simple zeros of the obtained averaged functions. Theoretically, we show that the maximum number of limit cycles of system \eqref{eq2.04} has no more than $kn_2$ (a rough bound) by using the $k$-th order averaging method. We believe that the first averaged function $f_k$ which is not identically zero is a polynomial in $r$ with odd terms. However, we cannot prove this, we leave this as a future research problem.

We remark that, though in the present paper, we focus our attention on the study of bifurcation of limit cycles of the continuous differential system \eqref{eq2.04}, the developed algorithmic approach admits a generalization to the case of studying the bifurcation of limit cycles for discontinuous differential systems. It is of great interest to employ our approach to analyze the bifurcation of limit cycles for differential systems in many different fields (biology, chemistry, economics, engineering, mathematics, physics, etc.). It will be beneficial to generalize our approach to the case of higher dimension differential systems by using the general form of the averaging method. We leave this as the future research problems. Furthermore, how to simplify and optimize the steps of the computations of the averaged functions is also worthy of further study.

%
\begin{acks}
Huang's work is partially supported by China Scholarship Council under Grant No.:~201806020128. Yap's work is partially supported by NSF Grants \#CCF-1423228 and \#CCF-1564132, and also a Chinese Academy of Science (Beijing) President's International Fellowship Initiative (2018), and Beihang International Visiting Professor Program No. Z2018060. The first author is grateful to Professor Dongming Wang for his encouragement and helpful suggestions, and to Chee Yap for inviting him to visit NYU Courant. Both authors thank the anonymous referees for their valuable comments on improving the presentation.
\end{acks}

%
\bibliographystyle{ACM-Reference-Format}
\balance
\bibliography{ref}

%
\appendix
\newpage
\section{Proof of Theorem \ref{tt1}}\label{A}
We first give some lemmas before we prove the Theorem \ref{tt1}. The following lemma plays a key role in determining the numbers $\nu_i$ and $N_i$.
\begin{lemma}\label{lemA1}
If $k\geq2$ and $p\in\mathbb{N}$, then for any polynomial $g_{n}(x)$ of degree $n$, \[\left(\frac{g_{n}(x)}{x^{k-1}}\right)^{(p)}=\frac{\bar{g}_{n}(x)}{x^{k+p-1}},\]
where $\bar{g}_{n}(x)$ is a polynomial of degree no more than $n$. Here $g^{(p)}$ denotes the $p$-order derivative of a function $g$.
\end{lemma}
\begin{proof}
The lemma follows directly from the following equality
\[\left(\frac{x^q}{x^{k-1}}\right)^{(p)}=(q-k+1)(q-k)\cdots(q-k+2-p)\frac{x^q}{x^{k+p-1}},\quad q\in\mathbb{N}.\]
\end{proof}

The lemma described below can be used to determine the expression form of the averaged function $f_i(r)$ in Theorem \ref{tt1}.
\begin{lemma}\label{lemA2}
Define the integral function
\begin{equation}\label{e3.9}
\begin{split}
M_{i,j,k}=\int_0^{\theta}s^i\sin^js\cos^ks ds,\quad i,j,k\in \mathbb{N}_+.
\end{split}
\end{equation}
Then we have the following recursive formula for $M_{i,j,k}$
\begin{equation}\label{e3.11}
\begin{split}
(j+1)M_{i,j,k}&=\theta^i\sin^{j+1}\theta\cos^{k-1}\theta-iM_{i-1,j+1,k-1}\\
&\quad+(k-1)M_{i,j+2,k-2}.
\end{split}
\end{equation}
Moreover, when $k=0$, we have
\begin{equation}\label{e3.14}
\begin{split}
j^2M_{i,j,0}&=-j\theta^i\sin^{j-1}\theta\cos\theta+i\theta^{i-1}\sin^j\theta\\
&+j(j-1)M_{i,j-2,0}-i(i-1)M_{i-2,j,0}.
\end{split}
\end{equation}

\end{lemma}

\begin{proof}
Doing integration by parts for \eqref{e3.9}, we have
\begin{equation}\label{e3.10}
\begin{split}
M_{i,j,k}&=\int_0^{\theta}s^i\sin^js\cos^{k-1}s(\sin s)' ds\\
&=\theta^i\sin^{j+1}\theta\cos^{k-1}\theta-\int_0^{\theta}\sin s\left(s^i\sin^js\cos^{k-1}s\right)'ds\\
&=\theta^i\sin^{j+1}\theta\cos^{k-1}\theta-iM_{i-1,j+1,k-1}\\
&\quad-jM_{i,j,k}+(k-1)M_{i,j+2,k-2}.\nonumber
\end{split}
\end{equation}
Then we find the recursive integral formula \eqref{e3.11}.

When $k=0$, doing integration by parts for $M_{i,j,0}$, we obtain
\begin{equation}\label{e3.12}
\begin{split}
M_{i,j,0}&=-\int_0^{\theta}s^i\sin^{j-1}s(\cos s)'ds\\
&=-\theta^i\sin^{j-1}\theta\cos\theta+(j-1)M_{i,j-2,0}\\
&\quad-(j-1)M_{i,j,0}+iM_{i-1,j-1,1}.
\end{split}
\end{equation}
On the other hand, by doing integration by parts for $M_{i-1,j-1,1}$, in a similar way we have
\begin{equation}\label{e3.13}
\begin{split}
jM_{i-1,j-1,1}=\theta^{i-1}\sin^j\theta-(i-1)M_{i-2,j,0}.
\end{split}
\end{equation}
Using equations \eqref{e3.12} and \eqref{e3.13}, we obtain \eqref{e3.14}.

\end{proof}

{\bf Proof of Theorem \ref{tt1}.}
Now refer to system \eqref{eq2.04}, we define the perturbed terms
\begin{equation}
\begin{split}
\tilde{p}_j(x,y)=\sum_{t=0}^{n_2}p_t^{j}(x,y),\quad
\tilde{q}_j(x,y)=\sum_{t=0}^{n_2}q_t^{j}(x,y),\nonumber
\end{split}
\end{equation}
with $p_t^{j}$, $q_t^{j}$ homogeneous polynomials of degree $t$. The change of coordinates \[x=\varepsilon X,\quad y=\varepsilon Y\]
carries system \eqref{eq2.04} into
\begin{equation}\label{e3.1}
\begin{split}
\dot{X}&=-Y+\sum_{m=2}^{n_1}\varepsilon^{m-1}P_m(X,Y)+\sum_{j=1}^k\sum_{t=0}^{n_2}\varepsilon^{j+t-1}p_t^{j}(X,Y),\\
\dot{Y}&=X+\sum_{m=2}^{n_1}\varepsilon^{m-1}Q_m(X,Y)+\sum_{j=1}^k\sum_{t=0}^{n_2}\varepsilon^{j+t-1}q_t^{j}(X,Y).
\end{split}
\end{equation}
In polar coordinates $X=rC$ and $Y=rS$ with $C=\cos\theta$, $S=\sin\theta$, system \eqref{e3.1} has the form
\begin{equation}
\begin{split}
\dot{r}=\frac{X\dot{X}+Y\dot{Y}}{r}\Big|_{X=rC,Y=rS},\quad
\dot{\theta}=\frac{X\dot{Y}-Y\dot{X}}{r^2}\Big|_{X=rC,Y=rS}.\nonumber
\end{split}
\end{equation}
Then
\begin{equation}\label{e3.2}
\begin{split}
\frac{dr}{d\theta}=r\frac{X\dot{X}+Y\dot{Y}}{X\dot{Y}-Y\dot{X}}\Big|_{X=rC,Y=rS}
=\frac{H_1(r,C,S,\varepsilon)}{r+H_2(r,C,S,\varepsilon)},
\end{split}
\end{equation}
where
\begin{equation}\label{e3.3}
\begin{split}
H_1(r,C,S,\varepsilon)&=\sum_{m=2}^{n_1}\varepsilon^{m-1}r^{m+1}[P_m(C,S)C+Q_m(C,S)S]\\
&+\sum_{j=1}^k\sum_{t=0}^{n_2}\varepsilon^{j+t-1}r^{t+1}[p_t^j(C,S)C+q_t^j(C,S)S],\\
H_2(r,C,S,\varepsilon)&=\sum_{m=2}^{n_1}\varepsilon^{m-1}r^{m}[Q_m(C,S)C-P_m(C,S)S]\\
&+\sum_{j=1}^k\sum_{t=0}^{n_2}\varepsilon^{j+t-1}r^{t}[q_t^j(C,S)C-p_t^j(C,S)S].\nonumber
\end{split}
\end{equation}
Computing the first-order Taylor expansion of $dr/d\theta$ in $\varepsilon$ we obtain
\[F_0=\frac{r[p_0^1(C,S)C+q_0^1(C,S)S]}{r+q_0^1(C,S)C-p_0^1(C,S)S}.\]
Since we assume that $F_0=0$, we need to let $p_0^1=q_0^1=0$. Then the resulting expression of $dr/d\theta$ is of the form
\begin{equation}\label{e3.4}
\begin{split}
\frac{dr}{d\theta}=\frac{B_1(r,C,S)\varepsilon+\cdots+B_{n_2+k-1}(r,C,S)\varepsilon^{n_2+k-1}}
{r+A_1(r,C,S)\varepsilon+\cdots+A_{n_2+k-1}(r,C,S)\varepsilon^{n_2+k-1}},
\end{split}
\end{equation}
where
\begin{equation}\label{e3.5}
\begin{split}
A_1(r,C,S)&=r^2[Q_2(C,S)C-P_2(C,S)S]\\
&\quad+\sum_{t=0}^1r^t[q_t^{2-t}(C,S)C-p_t^{2-t}(C,S)S],\\
B_1(r,C,S)&=r^3[P_2(C,S)C+Q_2(C,S)S]\\
&\quad+\sum_{t=0}^1r^{t+1}[p_{t}^{2-t}(C,S)C+q_{t}^{2-t}(C,S)S]\nonumber
\end{split}
\end{equation}
and the expressions of $A_i$ and $B_i$ for $i=2,\ldots,n_2+k-1$ are summation of a kind of polynomial functions in the form $r^{i_1}\bar{H}_{i_1}(C,S)$ with $i_1$ non-negative integer and $\bar{H}_{i_1}$ polynomial function in the variables $C$ and $S$. Moreover, by observing \eqref{e3.2} we know that $B_i$ for $i=2,\ldots,k-1$ is a polynomial in $r$ of degree at most $n_2+1$ without constant term; and $A_i$ is a polynomial in $r$ of degree at most $n_2$ in the form:
\begin{equation}\label{e3.5.0}
\begin{split}
A_i(r,C,S)=\bar{A}_{i,0}(C,S)+\bar{A}_{i,1}(r,C,S),\quad i=2,\ldots,k-1,
\end{split}
\end{equation}
where
$\bar{A}_{i,1}(r,C,S)$ is a polynomial in $r$ of degree at most $n_2$ without constant term, and
\begin{equation}\label{e3.5.1}
\begin{split}
\bar{A}_{i,0}(C,S)=q_0^{i+1}(C,S)C-p_0^{i+1}(C,S)S,\quad i=2,\ldots,k-1.
\end{split}
\end{equation}

We recall that, given any real value $|\eta|<1$, the following expansion holds:
\[\frac{1}{1+\eta}=\sum_{h_1\geq0}(-1)^{h_1}\eta^{h_1}.\]
Thus, equation \eqref{e3.4} can be written as
\begin{equation}\label{e3.6}
\begin{split}
\frac{dr}{d\theta}&=\left(\sum_{h_2=1}^{n_2+k-1}\frac{B_{h_2}}{r}\varepsilon^{h_2}\right)
\left[1+\sum_{h_1\geq1}(-1)^{h_1}\left(\sum_{h_2=1}^{n_2+k-1}\frac{A_{h_2}}{r}\varepsilon^{h_2}\right)^{h_1}\right]\\
&=\left(\sum_{h_2=1}^{k}\frac{B_{h_2}}{r}\varepsilon^{h_2}\right)\times\Bigg[1-\left(\sum_{h_2=1}^{k-1}\frac{\bar{A}_{h_2,0}+\bar{A}_{h_2,1}}{r}\varepsilon^{h_2}\right)+\cdots\\
&+(-1)^{k-1}\left(\sum_{h_2=1}^{k-1}\frac{\bar{A}_{h_2,0}+\bar{A}_{h_2,1}}{r}\varepsilon^{h_2}\right)^{k-1}\Bigg]+\varepsilon^{k+1}R(\varepsilon,C,S),\\
&=\sum_{i=1}^k\varepsilon^iF_i(r,C,S)+\mathcal{O}(\varepsilon^{k+1}),
\end{split}
\end{equation}
where
\begin{equation}\label{e3.7}
\begin{split}
F_1&=\frac{B_1}{r},\quad F_2=\frac{rB_2-A_1B_1}{r^2},\\
F_3&=\frac{r^2B_3-rA_1B_2-rA_2B_1+A_1^2B_1}{r^3}\nonumber
\end{split}
\end{equation}
and the expressions of $F_i$ for $i=4,\ldots,k$ are linear combination of a kind of functions in the form $r^{\alpha_1}A_{j_1}^{\alpha_2}B_{j_2}$ with $-i\leq\alpha_1\leq-1$, $1\leq j_1,\alpha_2\leq i-1$, and $1\leq j_2\leq i$ (here we have avoided the dependence on $(r,C,S)$ to simplify the notation). Recalling the property that $B_i$ is a polynomial in $r$ of degree at most $n_2+1$ without constant term and $A_i$ is a polynomial in $r$ of degree at most $n_2$ with constant term, we find that $F_1$ is a polynomial in $r$ of degree at most $n_2\geq2$ and $F_i$ is a rational function in $r$ of the form
\begin{equation}\label{e3.7.1}
\begin{split}
F_i=\bar{F}_i(r,C,S)/r^{i-1},\quad i=2,\ldots,k,
\end{split}
\end{equation}
where $\bar{F}_i(r,C,S)$ is a polynomial in $r$ of degree at most $in_2$.

In what follows, we first prove that there exist a non-negative integer $\nu_i$ and a polynomial function $\bar{f}_i(r)=\sum_{j=0}^{N_i}c_jr^j$, such that $r^{\nu_i}f_i(r)=\bar{f}_i(r)$ for $i=1,\ldots,k$, then we provide the bounds for the numbers $\nu_i$ and $N_i$.

Let $R^{SC}=\{r^{\lambda_1}\sin^{\lambda_2}\theta\cos^{\lambda_3}\theta: \lambda_1\in\mathbb{Z},\lambda_2,\lambda_3\in\mathbb{N}\}$ be a set of functions. It is obvious that each $F_i$ in \eqref{e3.6} (or \eqref{e3.7.1}) is a function generated by linear combination of elements of $R^{SC}$. Note that the explicit expression of $F_1$ is of the form
\begin{equation}\label{e3.8}
\begin{split}
F_1(r,C,S)&=r^2[P_2(C,S)C+Q_2(C,S)S]\\
&\quad+\sum_{t=0}^1r^{t}[p_{t}^{2-t}(C,S)C+q_{t}^{2-t}(C,S)S].\nonumber
\end{split}
\end{equation}
Now refer to \eqref{eq3.3.0}, it is easy to check that $y_1(\theta,r)$ is a function generated by linear combination of elements of the set of functions in the form $\{\theta r,r^{j_1}\sin^{j_2}\theta\cos^{j_3}\theta\}$ with $0\leq j_1\leq2$ and $0\leq j_2,j_3\leq3$.

Let $\bar{R}=R^{SC}\times\Theta=\{\theta^{\lambda_0}r^{\lambda_1}\sin^{\lambda_2}\theta\cos^{\lambda_3}\theta\}$ be a set of functions with $\lambda_0\in \mathbb{N}_+$. We denote by $\mbox{Span}(\bar{R})$ be the set of functions generated by linear combination of elements of $\bar{R}$. Next, we will show that the integral function $y_i(\theta,r)\in \mbox{Span}(\bar{R})$ for $i=2,\ldots,k$.

First, it is critical to observe that, the resulting form of $\partial^LF_i(\theta,r)$ is a function generated by linear combination of elements of $R^{SC}$. Since $y_1(\theta,r)$ contains $\theta$, the function in the square bracket of \eqref{eq3.3.0} is in $\mbox{Span}(\bar{R})$. In order to prove $y_i(\theta,r)\in \mbox{Span}(\bar{R})$, we need to consider the following integral equation:
\begin{equation}
\begin{split}
M_{i,j,k}=\int_0^{\theta}s^i\sin^js\cos^ks ds,\quad i,j,k\in \mathbb{N}_+.\nonumber
\end{split}
\end{equation}
Second, we claim that $M_{i,j,k}\in\mbox{Span}(\bar{R})$.
It follows from Lemma \ref{lemA2} that $M_{i,j,k}\in\mbox{Span}(\bar{R})$ if and only if $M_{i-1,j+1,k-1}\in\mbox{Span}(\bar{R})$ and $M_{i,j+2,k-2}\in\mbox{Span}(\bar{R})$.

Reuse the recursive formula \eqref{e3.11} until the subscript $i=0$ or $k=0$. In this way it suffices to consider $M_{0,j,k}\in\mbox{Span}(\bar{R})$ and $M_{i,j,0}\in\mbox{Span}(\bar{R})$. It is easy to judge that $M_{0,j,k}\in\mbox{Span}(\bar{R})$, so we focus on the proof of $M_{i,j,0}\in\mbox{Span}(\bar{R})$.
By using Lemma \ref{lemA2} and reusing the recursive formula \eqref{e3.14}, we conclude that $M_{i,j,0}\in\mbox{Span}(\bar{R})$. Then $M_{i,j,k}\in\mbox{Span}(\bar{R})$. Thus the desired result $y_i(\theta,r)\in \mbox{Span}(\bar{R})$ holds.

Finally, letting $\theta=2\pi$ in $y_i(\theta,r)$ (equation \eqref{eq2.01.2}) and taking into accounting the following formulae
\begin{equation}\label{e3.15}
\begin{split}
\cos(2\pi)=1,\quad \sin(2\pi)=0,\nonumber
\end{split}
\end{equation}
we prove that there exist a non-negative integer $\nu_i$ and a polynomial function $\bar{f}_i(r)=\sum_{j=0}^{N_i}c_jr^j$, such that $r^{\nu_i}f_i(r)=\bar{f}_i(r)$ for $i=1,\ldots,k$.

Next we provide the bounds for the numbers $\mu_i$ and $N_i$.

Case $i=1$, since $F_1$ is a polynomial in $r$ of degree 2 (at most $n_2$), we have by equation \eqref{eq3.3.0} that $y_1$ is a polynomial in $r$ of degree 2 (at most $n_2$).

We assume, by the induction hypothesis, that $y_i$ is a rational function in $r$ of the form
\begin{equation}\label{e3.7.2}
\begin{split}
y_{i-1}=\bar{y}_{i-1}(r,C,S)/r^{i-2},\quad i=2,\ldots,k,
\end{split}
\end{equation}
where $\bar{y}_{i-1}(r,C,S)$ is a polynomial in $r$ of degree at most $(i-1)n_2$.

In the expression of $y_i$ given in \eqref{eq3.3.0}, there only appear the previous functions $y_j$, for $1\leq j\leq i-1$. Now by using equation \eqref{e3.7.1} and Lemma \ref{lemA1}, for a given integer $\ell$ with $1\leq\ell\leq i-1$, we have the following summation function
\begin{equation}\label{e3.7.3}
\begin{split}
&\sum_{S_{\ell}}\partial^LF_{i-\ell}(\theta,r)\prod_{j=1}^{\ell}y_j(\theta,r)^{b_j}\\
&=\sum_{S_{\ell}}\frac{\bar{F}_{i-\ell}}{r^{i-\ell-1+L}}y_1^{b_1}y_2^{b_2}\cdots y_{\ell}^{b_{\ell}},\\
&=\sum_{S_{\ell}}\frac{\bar{F}_{i-\ell}}{r^{i-\ell-1+L}}\bar{y}_{1}^{b_1}\left(\frac{\bar{y}_{2}}{r}\right)^{b_2}\cdots\left(\frac{\bar{y}_{\ell}}{r^{\ell-1}}\right)^{b_{\ell}},\\
&=\sum_{S_{\ell}}\frac{\bar{F}_{i-\ell}}{r^{i-1}}\bar{y}_{1}^{b_1}\bar{y}_{2}^{b_2}\cdots\bar{y}_{\ell}^{b_{\ell}}.
\end{split}
\end{equation}
We have used the equalities $L=b_1+b_2+\cdots+b_{\ell}$ and $b_1+2b_2+\cdots+\ell b_{\ell}={\ell}$ to simplify \eqref{e3.7.3}. Combining equations \eqref{e3.7.1} and \eqref{e3.7.2}, we know that the numerator of the expression \eqref{e3.7.3} is a polynomial in $r$ with degree at most
\[(i-\ell)n_2+n_2(b_1+2b_2+\cdots+\ell b_{\ell})=in_2.\]
Thus $y_i$ is a rational function in $r$ of the form
\begin{equation}\label{e3.7.4}
\begin{split}
y_{i}=\bar{y}_{i}(r,C,S)/r^{i-1},\quad i=1,\ldots,k.
\end{split}
\end{equation}
where $\bar{y}_{i}(r,C,S)$ is a polynomial in $r$ of degree at most $in_2$.

Herewith, we prove that there exist a non-negative integer $\nu_i\leq i-1$ and a polynomial function $\bar{f}_i(r)=\sum_{j=0}^{N_i}c_jr^j$ with $N_i\leq in_2$, such that $r^{\nu_i}f_i(r)=\bar{f}_i(r)$ for $i=1,\ldots,k$.

Next we will show that the coefficients $c_j$ of $\bar{f}_i(r)$ is a polynomial in $\pi$ of degree at most $i$. To do this, we just consider the dependence on $\theta$ for brevity.

We define $\Delta Y_i=\{\theta^{\Delta_i}\sin^{k_1}\theta\cos^{k_2}\theta: 0\leq\Delta_i\leq i,k_1,k_2\in\mathbb{N}\}$ be a set of functions.
We claim that the following property holds
\begin{equation}\label{e3.7.5}
\begin{split}
y_{i}(\theta)\in\mbox{Span}(\Delta Y_i),\quad i=1,\ldots,k.
\end{split}
\end{equation}
We begin to prove this by induction.

Case $i=1$, we recall that $y_1(\theta)$ is a function generated by linear combination of elements of the set of functions $\{\theta,\sin^{j_2}\theta\cos^{j_3}\theta\}$ with $0\leq j_2,j_3\leq3$. It obvious that $y_1(2\pi)$ is a polynomial in $\pi$ of degree at most $1$.

Suppose that for $\bar{k}\leq i-1$, property \eqref{e3.7.5} holds, then for $\bar{k}=i$, using the integral equation \eqref{eq3.3.0}, for a given integer $\ell$ with $1\leq\ell\leq i-1$, we have
\begin{equation}\label{e3.7.6}
\begin{split}
\sum_{S_{\ell}}\prod_{j=1}^{\ell}y_j(\theta)^{b_j}=\sum_{S_{\ell}}y_1(\theta)^{b_1}y_2(\theta)^{b_2}\cdots y_{\ell}(\theta)^{b_{\ell}}.
\end{split}
\end{equation}
Note that by the induction hypothesis, we have $y_j(\theta)\in\mbox{Span}(\Delta Y_j)$ for $1\leq j\leq\ell\leq i-1$. Then the degree of $\theta$ in \eqref{e3.7.6} is at most $b_1+2b_2+\cdots+\ell b_{\ell}=\ell\leq i-1$.

By using Lemma \ref{lemA2} (the integral formulae \eqref{e3.11} and \eqref{e3.14}) and noting also that $\int_0^{\theta}\theta^{i-1}d\theta=\theta^i/i$, we find that $y_i(\theta)\in\mbox{Span}(\Delta Y_i)$. Finally, letting $\theta=2\pi$, we prove that $y_i(2\pi)$ is a polynomial in $\pi$ of degree at most $i$. That is to say, the coefficients $c_j$ of $\bar{f}_i(r)$ is a polynomial in $\pi$ of degree at most $i$.

Up to now, we finish the proof of Theorem \ref{tt1}.

\section{Fifth Order Averaging Formulae}\label{B}
We present some formulas computed by
{\bf Averformula} (Algorithm 2).
\begin{equation}\label{B01}
\begin{split}
Y_1&=\Big\{\int_0^{\theta}F_1(s,z)ds,\Big[\int_0^{\theta}F_1d\theta,\int_0^{2\pi}F_1d\theta\Big]\Big\},\\
Y_2&=\Big\{\int_0^{\theta}\Big(2F_2(s,z)+2\frac{\partial F_1(s,z)}{\partial z}\Big)y_1(s,z)ds,\\
&\Big[\int_0^{\theta}2\Big(F_2+\frac{\partial F_1}{\partial r}y_1\Big)d\theta,\int_0^{2\pi}\Big(F_2+\frac{\partial F_1}{\partial r}y_1\Big)d\theta\Big]\Big\}.
\end{split}
\end{equation}

\begin{equation}\label{BB1}
\begin{split}
y_k(\theta,z)=\int_0^{\theta}\mathbb{F}_k(s,z)ds,~\mbox{for}~k=1,\ldots,5,
\end{split}
\end{equation}
where
\begin{equation}
\begin{split}
\mathbb{F}_1(s,z)&=F_1(s,z),\\
\mathbb{F}_2(s,z)&=2F_2(s,z)+2\frac{\partial F_1(s,z)}{\partial z}y_1(s,z),\\
\mathbb{F}_3(s,z)&=6F_3(s,z)+6\frac{\partial F_2(s,z)}{\partial z}y_1(s,z)\\&\quad+3\frac{\partial F_1(s,z)}{\partial z}y_2(s,z)+3\frac{\partial^2 F_1(s,z)}{\partial z^2}y_1(s,z)^2,\\
\mathbb{F}_4(s,z)&=24F_4(s,z)+24\frac{\partial F_3(s,z)}{\partial z}y_1(s,z)\\
&\quad+12\frac{\partial F_2(s,z)}{\partial z}y_2(s,z)+12\frac{\partial^2 F_2(s,z)}{\partial z^2}y_1(s,z)^2\\
&\quad+4\frac{\partial F_1(s,z)}{\partial z}y_3(s,z)+12\frac{\partial^2 F_1(s,z)}{\partial z^2}y_1(s,z)y_2(s,z)\\
&\quad+4\frac{\partial^3 F_1(s,z)}{\partial z^3}y_1(s,z)^3,\\
\mathbb{F}_5(s,z)&=120F_5(s,z)+120\frac{\partial F_4(s,z)}{\partial z}y_1(s,z)\\
&\quad+60\frac{\partial F_3(s,z)}{\partial z}y_2(s,z)+60\frac{\partial^2 F_3(s,z)}{\partial z^2}y_1(s,z)^2\\
&\quad+20\frac{\partial F_2(s,z)}{\partial z}y_3(s,z)+60\frac{\partial^2 F_2(s,z)}{\partial z^2}y_1(s,z)y_2(s,z)\\
&\quad+20\frac{\partial^3 F_2(s,z)}{\partial z^3}y_1(s,z)^3+5\frac{\partial F_1(s,z)}{\partial z}y_4(s,z)\\
&\quad+15\frac{\partial^2 F_1(s,z)}{\partial z^2}y_2(s,z)^2\\
&\quad+20\frac{\partial^2 F_1(s,z)}{\partial z^2}y_1(s,z)y_3(s,z)\\
&\quad+30\frac{\partial^3 F_1(s,z)}{\partial z^3}y_1(s,z)^2y_2(s,z)\\
&\quad+5\frac{\partial^4 F_1(s,z)}{\partial z^4}y_1(s,z)^4.\\
\nonumber
\end{split}
\end{equation}

\section{Quadratic Systems}\label{C}
This appendix is an overflow from Subsection \ref{sect4.3}. In this subsection we report some results on quadratic differential
systems with centers of the form
\begin{equation}\label{eq4.3.0}
\begin{split}
\dot{x}&=-y+a_{20}x^2+a_{11}xy+a_{02}y^2,\\ \dot{y}&=x+b_{20}x^2+b_{11}xy+b_{02}y^2.
\end{split}
\end{equation}
\subsection{Isochronous Quadratic Centers}\label{C.1}
We recall that the classification of such quadratic system having an isochronous center at the origin is due to Loud \cite{wsl64}. He proved that after an affine change of variables and a rescaling of time any quadratic isochronous center can be written as one of the following four systems.
\begin{equation}
\begin{split}
&\mathbb{S}_1: \quad\dot{x}=-y+x^2-y^2,\quad \dot{y}=x(1+2y),\\
&\mathbb{S}_2: \quad\dot{x}=-y+x^2,\quad \dot{y}=x(1+y),\\
&\mathbb{S}_3: \quad\dot{x}=-y-\frac{4}{3}x^2,\quad \dot{y}=x(1-\frac{16}{3}y),\\
&\mathbb{S}_4: \quad\dot{x}=-y+\frac{16}{3}x^2-\frac{4}{3}y^2,\quad \dot{y}=x(1+\frac{8}{3}y).\nonumber
\end{split}
\end{equation}
In the case of limit cycles bifurcating from the periodic orbits surrounding such quadratic isochronous centers, Chicone and Jacobs in \cite{cm91} proved that, under all quadratic polynomial perturbations, at most 1 limit cycle bifurcate from the periodic orbits of $\mathbb{S}_1$, and at most 2 limit cycles bifurcate from the periodic orbits of $\mathbb{S}_2$, $\mathbb{S}_3$ and $\mathbb{S}_4$. Iliev obtained in \cite{ii98} that the cyclicity of the period annulus surrounding the center $\mathbb{S}_1$ is also 2.

Here we focus on the study of limit cycles that bifurcate from such quadratic isochronous centers, and the perturbation terms in \eqref{eq2.04} are taken as follows:
\begin{equation}
\begin{split}
p_{\alpha}(x,y,\varepsilon)&=\sum_{s=1}^{8}\sum_{j=1}^2\sum_{i=0}^j\varepsilon^sc_{s,i,j-i}x^{i}y^{j-i},\\
q_{\beta}(x,y,\varepsilon)&=\sum_{s=1}^{8}\sum_{j=1}^2\sum_{i=0}^j\varepsilon^sd_{s,i,j-i}x^{i}y^{j-i}.\nonumber
\end{split}
\end{equation}
Since the calculations and arguments are quite similar to those used in the previous proofs, we just summarize our results in the following Table \ref{tb02}.

\begin{table}[h]
\caption{Number of limit cycles for quadratic isochronous centers}\label{tb02}
\begin{center}
\begin{tabular}{ccccc}
  \hline
  Averaging order & $\mathbb{S}_1$ & $\mathbb{S}_2$ & $\mathbb{S}_3$ & $\mathbb{S}_4$\\ \hline
   1,2 & 0 & 0 & 0 & 0\\ \hline
   3,4 & 1 & 1 & 1 & 1\\ \hline
   5 & 1 & 2 & 2 & 2\\ \hline
   6 & 2 & 2 & 2 & 2\\ \hline
   7 & 2 & 2 & 2 & 2\\ \hline
   8 & - & - & - & -\\ \hline
\end{tabular}
\end{center}
\end{table}

We remark that the computation of the 8-th order averaged functions would be too demanding (Maple was consuming too much of the CPU during a calculation). Since we are providing lower bounds for the maximum number of limit cycles that bifurcate from the origin of such quadratic systems, the results could be improved using higher orders of the averaging theorem. Thus, we have a conjecture that some of the numbers 2 obtained in Table \ref{tb02} may could be increased to 3 as Bautin \cite{nn54} proved that in a sufficiently small neighborhood $\Omega$ of a quadratic center, all sufficiently small quadratic perturbations of the given system have at most three limit cycles in $\Omega$, and that three arbitrarily small-amplitude limit cycles can be produced.

\subsection{Reversible System with Two Centers}\label{C.2}
Next, we study the following reversible quadratic system
\begin{equation}\label{eqc.1}
\begin{split}
\dot{x}=y+a_1xy,\quad \dot{y}=-x+x^2+a_4y^2,\\
\end{split}
\end{equation}
with two centers $(0,0)$ and $(1,0)$, where $a_1$ and $a_4$ are real coefficients satisfying $a_1<-1$ (\cite{ym12}, Theorem 1). The authors in \cite{ym12} proved that 3 limit cycles can bifurcate from the center $(0,0)$ under the case $a_4=(a_1-5)/3$ based on the Melnikov function method by adding perturbed terms $p(x,y,\varepsilon)=\varepsilon a_{10}x$ and $q(x,y,\varepsilon)=\varepsilon (b_{01}y+b_{11}xy)$ (see Section 3 of \cite{ym12} for more details). Here using the averaging method we study system \eqref{eqc.1} by choosing a similar kind of perturbed terms, and then we give some remarks on the relations between these two methods.

First, introducing $x=-\bar{x}$, $y=\bar{y}$ into \eqref{eqc.1} results in
\begin{equation}\label{eqc.2}
\begin{split}
\dot{\bar{x}}=-\bar{y}+a_1\bar{x}\bar{y},\quad \dot{\bar{y}}=\bar{x}+\bar{x}^2+a_4\bar{y}^2,\\
\end{split}
\end{equation}
which is similar to system \eqref{eq2.03}, now has centers $(0,0)$ and $(-1,0)$. We then consider the perturbations
\begin{equation}\label{eqc.3}
\begin{split}
\dot{\bar{x}}&=-\bar{y}+a_1\bar{x}\bar{y}+\sum_{s=1}^{10}\varepsilon^sc_{s,1,0}\bar{x},\\
\dot{\bar{y}}&=\bar{x}+\bar{x}^2+a_4\bar{y}^2+\sum_{s=1}^{10}\varepsilon^s(d_{s,0,1}\bar{y}+d_{s,1,1}\bar{x}\bar{y})
\end{split}
\end{equation}
of \eqref{eqc.2}, where $c_{s,i,j}$ and $d_{s,i,j}$ are real parameters for $s=1,\ldots,10$.

Computing the averaged functions under the case $a_4+1\neq0$, we obtain the expressions of $f_k$'s up to $k=5$ as follows:
\begin{equation}
\begin{split}
&f_1(r)=\pi r(c_{1,1,0}+d_{1,0,1}),\quad f_2(r)=\pi r(c_{2,1,0}+d_{2,0,1}),\\
&f_3(r)=\frac{\pi r}{4}\big(\bar{A}_2r^2+\bar{A}_0\big),\quad f_4(r)=\frac{\pi r}{4}\big(\bar{B}_2r^2+\bar{B}_0\big),\\
&f_5(r)=\frac{\pi r}{24(a_4+1)^2}\big(\bar{C}_4r^4+\bar{C}_2r^2+\bar{C}_0\big),\nonumber
\end{split}
\end{equation}
and for $k=6,\ldots,10$, we have
\begin{equation}\label{eqc.4}
\begin{split}
f_k(r)=\frac{\pi r}{24}\big(\bar{D}_{4,k}r^4+\bar{D}_{2,k}r^2+\bar{D}_{0,k}\big),
\end{split}
\end{equation}
where
\begin{equation}
\begin{split}
\bar{A}_2&=(a_1+2a_4)(a_1-a_4-1)c_{1,1,0}-(a_4+1)d_{1,1,1},\\ \bar{A}_0&=4(c_{3,1,0}+d_{3,0,1}),\\
\bar{B}_2&=(a_1+2a_4)(a_1-a_4-1)c_{2,1,0}-(a_4+1)d_{2,1,1},\\ \bar{B}_0&=4(c_{4,1,0}+d_{4,0,1}),\\
\bar{C}_4&=a_1(a_4+1)^2(a_1-a_4)(a_1+2a_4)(a_1-3a_4-5)c_{2,1,0},\\
\bar{C}_2&=-6a_1(a_1+2a_4)(a_1-a_4-1)(a_1+a_4-1)c_{1,1,0}^3\\
&\quad+6(a_4+1)^2(a_1+2a_4)(a_1-a_4-1)c_{3,1,0}\\
&\quad-6(a_4+1)^3d_{3,1,1},\\
\bar{C}_0&=24(a_4+1)^2(c_{5,1,0}+d_{5,0,1}),\\
\bar{D}_{4,k}&=a_1(a_1-a_4)(a_1+2a_4)(a_1-3a_4-5)c_{k-4,1,0},\\
\bar{D}_{2,k}&=6(a_1+2a_4)(a_1-a_4-1)c_{k-2,1,0}-6(a_4+1)d_{k-2,1,1},\\
\bar{D}_{0,k}&=24(c_{k,1,0}+d_{k,0,1}).\nonumber
\end{split}
\end{equation}
In view of these expressions of the obtained averaged functions, we find that the $k$-th ($k=5,\ldots,10$) order averaging provides the existence of at most two small-amplitude limit cycles of the perturbed system \eqref{eqc.3} and this number can be reached under the condition $(a_1-a_4)(a_1+2a_4)(a_1-3a_4-5)\neq0$. We conjecture that the maximal number of small-amplitude limit cycles of the perturbed system \eqref{eqc.3} is 2 up to the $k$-th order averaging for any $k\geq6$. This problem might be proved by using the recursive integral equation \eqref{eq3.3.0}.

Our result on the quadratic system \eqref{eqc.1} describes the different mechanisms between the averaging method and the Melnikov function method when studying the number of limit cycles that can appear in a Hopf bifurcation from centers. The number of limit cycles obtained by the averaging method in some cases (under a similar kind of perturbations) seems to be less than the number obtained by the Melnikov function method. We want to say that in the study of the limit cycles which bifurcate from a period annulus surrounding the center, the equivalence between the averaging method and the Melnikov function method at any order has been proved in \cite{ab17,mvx16}.

\end{document}